\documentclass{ecai} 

\usepackage{latexsym}
\usepackage{amssymb}
\usepackage{amsmath}
\usepackage{amsthm}
\usepackage{booktabs}
\usepackage{enumitem}
\usepackage{graphicx}
\usepackage{color}
\usepackage{algorithm}
\usepackage{pdfpages}
\usepackage{marginnote}
\usepackage[noEnd=true]{algpseudocodex}
\usepackage[switch]{lineno}
\usepackage{xcolor}
\usepackage{enumitem}
\usepackage{framed,graphicx,xcolor}

\newcommand{\subf}[2]{%
  {\small\begin{tabular}[t]{@{}c@{}}
  #1\\#2
  \end{tabular}}%
}

\newtheorem{theorem}{Theorem}
\newtheorem{lemma}[theorem]{Lemma}

\newtheorem{proposition}[theorem]{Proposition}

\newtheorem{definition}{Definition}
\newtheorem{example}{Example}
\newtheorem{observation}{Observation}

\newenvironment{proofidea}{\paragraph{Highlights of Proof Ideas:}}{\hfill$\square$}

\newcommand{\BibTeX}{B\kern-.05em{\sc i\kern-.025em b}\kern-.08em\TeX}

\begin{document}
\definecolor{shadecolor}{gray}{0.9}


\begin{frontmatter}


\paperid{1360} 


\title{Equilibria in multiagent online problems with predictions}


\author[A]{\fnms{Gabriel}~\snm{Istrate}\thanks{Corresponding Author. Email: gabriel.istrate@unibuc.ro.}}
\author[B]{\fnms{Cosmin}~\snm{Bonchi\c{s}}}
\author[B]{\fnms{Victor}~\snm{Bogdan}}

\address[A]{University of Bucharest}
\address[B]{West University of Timi\c{s}oara}


\begin{abstract}
We study the effect of \emph{predictions} on competitive ratio equilibria in multi-agent online problems, using a recently defined \emph{multi-agent ski-rental problem} as a testbed. We characterize predictionless equilibria, formally define equilibria with predictions, and show that (for multiagent ski-rental) they have properties reminscent of their predictionless counterparts. 

We also investigate how using an algorithm with predictions can benefit an individual agent when playing an equilibrium strategy profile for multiagent ski-rental. We give a theoretical characterization of the consistency and robustness achieved by the algorithm, and experimentally benchmark its average competitive ratio. 
\end{abstract}

\end{frontmatter}


\section{Introduction}

In many real-life scenarios  agents interact strategically, increasingly assisted by \emph{predictors} which provide forecasts of the behavior of the other agents (as well as on the optimal course of action to adopt). A convincing example is the use of navigation assistants:  in such apps the driver receives information about the current congestion in the network, as well as an estimate of the quickest route to the desired destination, given current traffic conditions. It is up to the individual agent to follow the recommendations of such an assistant: whereas such apps are  generally trustworthy, they may fail to display last-minute accidents that impact the choice of the optimal route. Similarly, the availability of pre-election polls may induce voters (in elections with more than two candidates) to behave strategically. Conversely, (doctored) opinion polls can been used strategically, aiming to help certain candidates by influencing voter turnout and behavior\footnote{This has recently happened in real-life elections in Romania. The strategic role of predictions in voting is a fascinating topic for further research.}. From an individual voter perspective one cannot, therefore, trust such "predictions", and must use a decision algorithm that is robust to the perceived risk of using untrusted polls. 

Given these considerations, it should come as no surprise that  augmenting algorithms by \emph{predictors}   \cite{algo-predictions} is one of the most exciting recent research directions in algorithmics. In the special case of online algorithms, a framework \cite{lykouris2018competitive,purohit2018improving} that already became standard benchmarks the quality of an algorithm by two competitive ratios: its \emph{consistency}, 
 the competitive ratio obtained for a perfect predictor, and its \emph{robustness}, defined as the competitive ratio for pessimal (even adversarial) predictors. The two measures are generally opposed, and algorithms have been proposed \cite{purohit2018improving} trading one for the other. 

Recently a proposal for exteding the analysis of online algorithms with preditions to a multiagent setting has been made \cite{anon2025}. A multiagent extension of the ski-rental problem was introduced, and studied assuming that each agent is endowed with  \textit{two predictors}: one for its (future) behavior, and one for the behavior of all other agents. The goal of \cite{anon2025} was to characterize the best possible competitive ratios attainable in the four cases, corresponding to each of these predictors being exact and pessimal, respectively. However, no attempt was made in \cite{anon2025} to actually define a game-theoretic notion of equilibrium appropriate for agents using algorithms with predictions.  

\noindent\textbf{Contributions.}  Our main conceptual contribution is \textbf{the  definition of a notion of equilibrium with predictions} for multiagent competitive online problems \cite{borodin2005online}. We investigate such equilibria in the multiagent ski-rental problem. We aim to answer the following questions:  
\vspace{-5mm}
\begin{shaded} 
\begin{itemize}[leftmargin=*]
\item[(1).] what properties do equilibria with (and without) predictions for the multiagent ski rental problem have?  what is the competitive ratio achieved by agents playing  such equilibria?
\item[(2).] (how) can algorithms employ predictions to improve their competitive ratio when playing an equilibrium strategy profile?
\end{itemize} 
\end{shaded} 
\vspace{-2mm}
We give the following answers to the questions above: 
\begin{shaded}
\begin{itemize} 
\item[-] We characterize  (Theorem~\ref{thm:predictionless})  (predictionless) competitive ratio equilibria, as well as (Theorem~\ref{thm:ptb}) those arising when agents predict others (and respond rationally), but don't self-predict.  
\item[-] We define equilibria with predictions and show (Theorem~\ref{thm:ptb2}) that they have properties reminiscent of their counterparts in the case with no self-predictions. 
\end{itemize} 
\end{shaded} 
We also give an answer to question (2) by studying (theoretically and experimentally) an algorithm (introduced in \cite{anon2025}) with a tunable parameter $\lambda=1$.  One can implement the equilibrium behavior of Theorem~\ref{thm:ptb} by having all agents use  Algorithm~\ref{zeroprime-alg} with parameter $\lambda =1$. We investigate the scenario when (unknown to other agents) one particular agent deviates from this equilibrium behavior by using a different value of $\lambda$. Our conclusions, in a nutshell, are as follows:  
\begin{shaded} 
\begin{itemize} 
\item[-] We can precisely characterize theoretically (Theorem~\ref{eq-one-prediction}) the worst-case consistency and robustness of Algorithm~\ref{zeroprime-alg} on all equilibria. Giving such guarantees for the algorithm was not feasible when run assuming random models for others' behavior \cite{anon2025}. In particular, for certain values of $\lambda$ \textit{using correct predictors provably helps in achieving a lower competitive ratio}. 
\item[-] When investigated experimentally, some of the equilibria turn out to be "trivial", in that the predictionless algorithm is able to find optimal (1-competitive) solutions for most instances. For such equilibria adding self-predictions is pointless.
\end{itemize} 
\end{shaded} 
\begin{shaded} 
\begin{itemize} 
\item[-] For "nontrivial" equilibria, the behavior of is Algorithm~\ref{zeroprime-alg} is qualitatively similar to that of the corresponding algorithm for the multiagent ski-rental with prediction problem \cite{purohit2018improving}: the performance of the algorithm with predictions is best for exact predictors, and degrades when prediction error increases. It seems to become more effective the later the day of license acquisition at equilibrium is. 
\end{itemize}
\end{shaded}


\noindent \textbf{Broader impact.} The classical notion of Nash equilibrium implicitly assumes that agents perfectly predict the behavior of the other agents, and best respond to this prediction. It is natural to inquire about a notion of "equilibrium with predictions" in settings where agents may be risk-averse (rather than risk neutral). Such agents are able to strategically exploit the predictor, obtaining better results when the predictor is correct, even in the presence of strategic behavior.  

Our paper provides a plausible first possible formulation for such a notion for competitive online games with predictions, and provides a proof-of-concept analysis for the newly introduced formulation that showcases its feasibility. It is our hope that the current problem brings attention to the important aspects of the interplay of prediction and optimization in multiagent systems, and that the formalism we introduce proves useful beyond multiagent ski-rental. 

More generally, we hope that our research brings more attention from the scientific community to the strategic (game-theoretic) aspects of algorithms with predictions, and that it is adapted to general non-cooperative models, as well as to more realistic versions of game-theoretic scenarios, such as \emph{dynamic games} \cite{bacsar1998dynamic}, in particular to open-loop models of such games \cite{fudenberg1988open}. 

\noindent \textbf{Outline. } The plan of the paper is as follows: in Section~\ref{ski-rental} we define the main problem we are interested in (as well as that of a generalization of the ordinary ski-rental problem that will be useful in its analysis). In Section~\ref{sec:prophet} we investigate a case we  consider as a baseline: that of equilibria without self-predictions.  We then situate equilibria with predictions in a decision-theoretic perspective and formally define them in Section~\ref{eq:predictions}. In  Section~\ref{sec:ec:predictions} we investigate the extent to which such equilibria resemble their predictionless counterparts. Finally, in Sections~\ref{sec:5} we benchmark the impact of using predictions on competitive ratio equilibria. 

\textbf{In this manuscript we will present the main ideas of our proofs}, relegating the full details to the Appendix. 

\section{Preliminaries} 

An \emph{online algorithm} $A$ takes its input $x$ from a set of possible inputs $\Omega$. Such inputs have a sequential nature, are not known from the outset, being disclosed as time goes by. 
Algorithm $A$ is called \emph{$\gamma$-competitive} iff there exists a real number $d$ (which may depend on $\gamma$, but not on $x$) such that $cost(A(x))\leq \gamma \cdot cost(OPT(x))+d$ for all $x$. When $d=0$ we call algorithm $A$ \textit{strongly $\gamma$-competitive}. By forcing somewhat the language, we refer to the quantity 
$\frac{cost(A(x))}{cost(OPT(x))}$ as the \emph{competitive ratio of algorithm $A$ on input $x$}, denoted by $c_{A}(p)$. We will also denote by $c_{OPT}(x)$ the optimal competitive ratio of an algorithm on input $x$. This assumes the existence of a class $\mathcal{A}$ of online algorithms from which $A$ is drawn, hence $c_{OPT}(x)=\inf(c_{A}(x):A\in \mathcal{A}).$  Other standard concepts in the competitive analysis of online algorithms are reviewed e.g. in \cite{borodin2005online}.

An online algorithm \emph{with predictions} uses as input both $x\in \Omega$ and a \emph{predictor $\widehat{Z}$} for some unknown quantity $Z(x)$ in a "prediction space" $\mathcal{P}$. 
We benchmark the performance of an 
an online algorithm $A$ with predictions on an instance $x$ as follows: instead of a single-argument function $A(x)$ we use a \emph{three-argument function} $A(x,Z,\widehat{Z})$. The following approach was proposed in \cite{purohit2018improving}: given a predictor $\widehat{Z}$ for the future, the performance of an online algorithm is measured not by one, but by a pair of competitive ratios: \textit{consistency},  the competitive ratio achieved in the case where the predictor is perfect, and \textit{robustness}, the competitive ratio in the case where the predictor is pessimal.  We stress that in \cite{purohit2018improving,gollapudi2019online,angelopoulos2020online,wei2020optimal} consistency/robustness are \textbf{not} aggregated into a single performance measure. Instead, one computes the Pareto frontier of the two measures, and/or optimizes consistency for a desired degree of robustness.

\begin{definition} An online algorithm with predictions $A$ is \emph{$\gamma$-robust  on input $x$} iff there exists a constant $d$ (that may depend on $\gamma$, but not on $\widehat{Z}$) such that for all values $\widehat{Z}$ of the predictor 
\begin{equation} 
cost(A(x,\widehat{Z}))\leq \gamma \cdot cost(OPT(x))+d
\end{equation} 
Strong $\gamma$-robustness is obtained in the special case when $d=0$. The algorithm is \emph{$\beta$-consistent} if 
 \begin{equation} 
 cost(A(x,\widehat{Z}=Z))\leq \beta\cdot cost(OPT(x)).
 \end{equation}
  If the guarantees hold for every input $x$ (with the same constant $d$) we will call the algorithm  (strongly) $\gamma$-robust and $\beta$-consistent. 
\label{def:one}
\end{definition} 


Consider a $n$-person competitive online game \cite{engelberg2016equilibria}. 
A \emph{strategy profile} in such a game is a vector of online algorithms $W=(W_1,W_2,\ldots, W_n)$, one for each player. Given strategy profile $W,$ denote by $BR_i(W)$ the set of strategy profiles where the $i$'th agent plays a best-response (i.e. an algorithm minimizing its competitive ratio) to the other agents playing $W$. A \textbf{competitive ratio equilibrium} is a strategy profile $(W_1,W_2,\ldots, W_n)$ such that each $W_i$ is a best-response to the programs of the other agents.

A \textbf{competitive online game with predictions} is the variant of the previous setting obtained by assuming that agents employ instead algorithms with predictions. We will consider such games under a \textbf{two-predictor model}: we  assume that every agent $i$ comes with two predictors - a \emph{self predictor} for its own future, and an \emph{others' predictor},  for the actions of other agents.
A formalization of this concept was sketched in \cite{anon2025}. The reader might want to consult this paper (although such a general definition is not really needed for understanding the results in this paper). 

\section{The Multiagent Ski-Rental Problem}
\label{ski-rental}
The main problem we are concerned with in this paper is the following \emph{multiagent ski-rental problem}: 

\begin{definition}[Multiagent Ski Rental \cite{anon2025}] 
$n$ agents are initially active but may become inactive.\footnote{once an agent becomes inactive it will be inactive forever.} Active agents need a resource for their daily activity. Each day,  active agents have the option to (individually) rent the resource, at a cost of 1\$/day. They can also cooperate in order to to buy a group license that will cost $B>1$ dollars. For this, each agent $i$ may pledge some 
amount $w_r>0$ or \emph{refrain from pledging} (equivalently, pledge 0).\footnote{We assume that both pledges $w_r$ and the price $B$ are integers.} If the total sum pledged is at least $B$ then the group license is acquired,\footnote{When the license is overpledged agents pay their pledged sums.} and the use of the resource becomes free for all remaining active agents from that moment on. Otherwise (we call such pledges \emph{inconsequential}) the pledges are nullified. Instead, every active agent must (individually) rent 
the resource for the day.  Agents are strategic, in that they care about their overall costs. They are faced, on the other hand, with deep uncertainty concerning the number of days they will be active. So they choose instead to minimize their competitive ratio, rather than their total cost. 
\end{definition}

We will also use the following variant of classical ski-rental: 

\begin{definition}[Ski Rental With Known Varying Prices] An agent is facing a ski-rental problem where the buying price varies from day to day (as opposed to the cost of renting, which is always 1\$). Denote by $0\leq p_i\leq B$ the cost of buying skis \textbf{exactly on day $i$}, and let $P_i = i-1+p_i$ be the \textbf{total cost if buying exactly on day $i$} (including the cost of renting on the previous $i-1$ days). 
Finally, let $p=(p_{i})_{i\geq 1}$. We truncate $p$ at the first $i$ such that $p_i=0$ (such an $i$ will be called a \emph{free day}). Call a day $i$ such that $p_i=1$ a \emph{bargain day}. Define $
M_{*}(p)=min(P_i:i\geq 1)$ and, for $t\geq 1$, $Q_t= min(P_i:i\geq t)$.

\label{ski-rental-known-varying} 
\end{definition} 

The following classes of algorithms pertain to ski-rental with varying prices and multiagent ski-rental: 

\begin{itemize}[leftmargin=*] 
\item[-] A \emph{predictionless algorithm} is a function $f:\mathbb{N}\rightarrow \mathbb{N}$. $f(t)$ represents the amount pledged by the agent at time $t$. A predictionless algorithm uses no information: it does not employ any predictions (about self-behavior or the behavior of others). 
\item[-] An \emph{algorithm with (self and others) predictions} is a function $f(k,\widehat{T},\widehat{p})$. Here $k\geq 1$ is an integer, $\widehat{T}\geq 1$ is  a prediction of the true active time of the agent, and $\widehat{p}:\mathbb{N}\rightarrow \{0,\ldots, B\}$ is a prediction of the total amount pledged by the other agents, by day. 
\item[-] An algorithm with predictions is \emph{rational} if for every $k,\widehat{T},\widehat{p}$, the agent either pledges just enough to buy the license (according to the others' predictor) or pledges 0.
 \item[-] An algorithm with predictions is \emph{simple} if whenever $\widehat{T_1},\widehat{T_2}<M_{*}(\widehat{p})$ or $\widehat{T_1},\widehat{T_2}\geq M_{*}(\widehat{p})$ then $f(k,\widehat{T_1},\widehat{p})=f(k,\widehat{T_2},\widehat{p})$. That is, what matters for the prediction is whether or not $\widehat{T}\geq M_{*}(\widehat{p})$.
\end{itemize}

\section{Multiagent Ski-Rental: The Predictionless Case} 
\label{sec:prophet}
As a base case we investigate the case of predictionless algorithms, completely characterizing competitive ratio equilibria arising from such algorithms: 

\begin{theorem} The following are true: 
\begin{itemize} 
\item[(a).] In every competitive ratio equilibrium the license is eventually bought on some day.    
\item[(b).] Consider $n$ predictionless algorithms specified by functions $(f_1,f_2,\ldots, f_n)$, and for $1\leq i\leq n$ let \[
Z_{i}=min(k-1+max(B-\sum_{j\neq i}f_{j}(k),0): 1\leq k\leq B).\]
 \textbf{Then} 
\noindent \textbf{$(f_1,f_2,\ldots, f_n)$ is a competitive ratio equilibrium iff:} 
\begin{itemize} 
\item[(i).] There exists $r$ s.t. $\sum_{j=1}^{n} f_{j}(r)\geq B$. W.l.o.g. denote by $r$ the smallest such day. Then, in fact, we have $\sum_{j=1}^{n} f_{j}(r)= B$. 

\item[(ii).] for all $i$ (if any) such that $\sum_{j\neq i} f_{j}(r)=B$ and $r\leq Z_i+1$, it holds that in fact $r=Z_i+1$ and $f_i(r)=0$.  

\item[(iii).] for all $i$  s.t.  $\exists k\leq min(r,Z_i)$ with $\sum_{j\neq i}f_{j}(k)=B-1$,  either ($k=r=Z_i$ and $f_i(r)=1$) or ($k=Z_i=r-1$, $f_{i}(r-1)=f_{i}(r)=0$, and $\sum_{j\neq i}f_{j}(k)=B$). 

\item[(iv).] for all other $i$ and all $j\neq r$,  
\[
\sum_{k\neq i} f_{k}(j)\leq B+j-1- min(j,Z_i)\frac{f_i(r)+r-1}{min(r,Z_i)}.
\]
\end{itemize} 
\end{itemize} 
\label{thm:predictionless} 
\end{theorem} 

\begin{proofidea} We will make use of the following characterization of best-response in the Ski-Rental Problem with Varying Prices (Theorem 4.4. in \cite{anon2025}, see that paper for a proof): 

\begin{proposition} Optimal competitive algorithms in Ski Rental with Known Varying Prices are: 
\begin{itemize}[leftmargin=*]
\item[(a).] if a free day exists on day $d=M_{*}(p)+1$ or a bargain day exists on day $d=M_{*}(p)$, then the agent that waits for the first such day, renting until then is 1-competitive. The only other 1-competitive algorithm exists when the first bargain day $b=M_{*}(p)$ is followed by the first free day $f=M_{*}(p)+1$; the algorithm rents on all days, waiting for day $f$.  
\item[(b).] Suppose case (a) does not apply, i.e. the daily prices are $p_1,p_2,\ldots, p_n \ldots $, $1< p_i\leq B$ for $1\leq i\leq M_{*}(p)$, $p_{M_{*}(p)+1}\geq 1$. Then no deterministic algorithm can be better than $c_{OPT}(p)$-competitive, where 
\begin{equation} 
 c_{OPT}(p):=
min(\{\frac{P_r}{r}:r\leq M_{*}(p)\}\cup\{ \frac{Q_{M_{*}(p)}}{M_{*}(p)}\}).
\end{equation} 
 Algorithms that buy on a day in $argmin(\{\frac{P_r}{r}:1\leq r\leq M_{*}(p)\}\cup \{\frac{P_r}{M_{*}(p)}: r\geq M_{*}(p), P_r=Q_{M_{*}(p)}\})$ are the only ones to achieve this ratio. 
\end{itemize} 
\label{thm-full-free1}
\end{proposition} 

We use this result as follows: We take the single agent perspective in multi-agent ski rental, and investigate the best response to the other agents' behavior. Fixing the algorithms of the other agents means that the agent knows \emph{how much money it would need to contribute on any given day to make the group license feasible}. This effectively reduces the problem to computing a best response to the ski-rental problem with known varying prices, as described in Definition~\ref{ski-rental-known-varying}.
\begin{itemize}[leftmargin=*] 
\item[(a).] Consider the perspective of agent 1. It will never have a free day, hence its best response is to pledge the required amount on its optimal day/first bargain day (if it exists), at the required cost.
\item[(b).]  

\begin{itemize}[leftmargin=*] 

\item[(i).] To have a competitive ratio equilibrium, every agent $i=1,\ldots, n$ has to play a best-response in the ski-rental problem with varying prices $p^i=(p_j)_{j\geq 1}$,  $p_j=max(B-\sum_{k\neq i}f_{k}(j),0)$.  By point (a), the license is bought on some day $r$. So $\sum_{j=1}^{n} f_{j}(r)\geq B$. In fact $\sum_{j=1}^{n} f_{j}(r)=B$: if it were not the case then any participating agent could lower its price paid on day $r$ (hence its competitive ratio) by pledging just enough to buy the license. 

\item[(ii.)] Consider  such an agent $i$. The hypothesis implies the fact that it will have a free day on day $r\leq M_i+1$. By Proposition~\ref{thm-full-free1}(a), it is optimal (1-competitive) for the agent $i$ not to pledge any positive amounts and simply wait for its free day.

\item[(iii).] Consider such an agent $i$. The hypothesis implies the fact that it has a bargain day on a day $k\leq min(r,M_i)$. Given that the license is bought on day $r$, it is either the case that the first such day is $k=r$ or that $k=r-1$ and $r$ is a free day for agent $i$. In this case the agent doesn't have to pledge on day $r-1$ since it gets the license for free on day $r$.

\item[(iv).] Consider now an agent $i$ such that $f_i(r)> 1$. From the agent's perspective $P_j=j-1+B-\sum_{k\neq i} f_{k}(j)$. In particular $M_{*}(p^i)\leq P_{1}\leq B$. 
The agent has no free/bargain days. We show (see the Appendix) that day $r$ is an optimal day to buy for agent $r$. 

\end{itemize} 
\end{itemize} 
\end{proofidea} 

The following result complements Theorem~\ref{thm:predictionless}, by computing the competitive ratios of agents in cases (ii),(iii),(iv) of the Theorem:  

\begin{theorem}
Consider a competitive ratio equilibrium as described in points (i)-(iv) from Theorem~\ref{thm:predictionless}. Then: 
\begin{itemize} 
\item[-] if $i$ is an agent in case (ii) then $c_{i}=1$. 
\item[-] if $i$ is an agent in case (iii) then $c_{i}=1$.
\item[-] if $i$ is an agent in case (iv) then $c_{i}=\frac{f_i(r)+r-1}{min(r,M_i)}$.
\end{itemize} 
\label{foo}
\end{theorem} 
\begin{proofidea} 
In the first case the agent has a free day. In the second one it has a bargain day. In the third case we use the fact that 
$OPT_r=min(r,M_i)$, while the agent rents for $r-1$ days and pledges $f_{i}(r)$ on day $r$, paying a total of $f_{i}(r)+r-1$ as a result. 
\end{proofidea} 

\subsection{Equilibria of Rational Agents with No Self-Predictions}
\label{no-self-pred}
We now move to the case where agents predict the others but don't self predict. We will make a further assumption, that make the equilibria especially simple: that all agents use rational algorithms\footnote{There is a subtle difference between the algorithms used in this section (whose rationality requires  predicting the aggregate behavior of other agents) and the predictionless algorithms of the previous section (that may stumble on the "correct pledge" by chance.} This is a reasonable assumption; it makes players not "gratuitously pledge" when they have no intention to buy. Equilibria, described below, involve two types of agents: 
\begin{itemize}[leftmargin=*]
\item[-] agents in some coalition $S$ coordinate their behaviors by pledging on the same day a total amount sufficient to buy the group licence. 
\item[-] agents outside of $S$ "freeride", by waiting (i.e. pledging 0 and renting) to get the group license for free\footnote{we can eliminate freeriding by slightly modifying the rules of the game, so that only pledging agents get the group license. With this change no agent faces a free day anymore.}. Formally: 
\end{itemize} 

\begin{theorem} The following are true: 

(a). All competitive ratio equilibria of \emph{rational algorithms with (perfect) others' predictors and pessimal/no self-predictors} are as follows: on some day $r\leq 2B-1$ agents in some subset $S\subseteq [n]$ pledge amounts $(w_i)_{i\in S}$, $\sum_{i\in S} w_i=B$, such that $1\leq w_i$ and for all $i\in S$, 
\begin{equation}
(w_i-1)\cdot min(B,r-1+w_i)\leq r(B-1)\mbox{ if }r\leq B,\mbox{ and }
\end{equation}  
\begin{equation} 
r-1+w_i\leq 2B-1\mbox{ if }B+1\leq r\leq 2B-1.
\end{equation}   
Agents don't pledge at any other times.   

(b). In the equilibria described at point (c) the competitive ratio of agent $i\in S$ is 
\[
1+\frac{w_i-1}{r}\mbox{ if }r\leq B,
\frac{r+w_i-1}{B}\mbox{ if }B+1\leq r\leq 2r-1.
\]
  The competitive ratio of an agent $i\not\in S$ is 1 if $r\leq B+1$, $\frac{r-1}{B}$ if $B+1\leq r\leq 2B-1$. 
\label{thm:ptb} 
\end{theorem}

\begin{proof} 
See the Appendix. 
\end{proof}

\section{Equilibria with Predictions: Decision-Theoretic  Perspective and Definition}
\label{eq:predictions}
Suppose $n$ agents play a given game. Each agent is endowed with some beliefs(predictions), possibly probabilistic,  on the behavior of the other players. Decision theory (see e.g. \cite{bonanno2017decision}) has considered several possible courses of action for an individual agent, including: 
\begin{itemize}[leftmargin=*]  
\item[-] \textbf{MaxiMin:} The agent chooses an action maximizing its payoff, \textit{assuming worst-case(adversarial) behavior of other agents}.
\item[-] \textbf{MaxiMax:} the agent chooses a course of action maximizing its payoff, \textit{assuming best-case behavior of other agents}. 
\item[-] \textbf{Best Response:} The agent chooses the action that maximizes its expected utility, given its beliefs about the actions of other players. When all players employ this strategy \textit{and are able to perfectly predict their opponents' behavior}, the resulting concept is the well-known Nash Equilibrium. 
\end{itemize}  
The first and third courses of action listed above represent two extremes in acting upon one's beliefs: in one case (MaxiMin) one completely disregards them and "act safely". In the other case, the appropriateness of the (best-response) action of a given agent is contingent on the exactness of its  prediction on the behavior of its opponents. In reality, agents may want to play strategies other than expected maximization - this is, for instance, the case if agents are \textit{risk-seeking}. Nash equilibria with different optimization strategies have been investigated (several relevant references include \cite{fiat2010players,slumbers2023game,mavronicolas2015minimizing,mavronicolas2020conditional}). However, sometimes one may want to combine several courses of action: for instance, the Hurwicz index of pessimism (e.g. \cite{bonanno2017decision}, Section 6.4) is defined as a weighted linear combination of MaxiMin and MaxiMax.  

A risk-reward framework for the analysis of online algorithms was proposed in \cite{al1999risk} (see also \cite{ding2005risk, ma2010optimal,ni2009online} for subsequent work using it), and rediscovered in an isomorphic formulation in \cite{lykouris2018competitive, purohit2018improving}, in the context of measuring the performance of online algorithms with predictions. Specifically, given a predictor $\widehat{T}$ for the future, the performance of an online algorithm is measured not by one, but by a pair of real numbers: \textit{consistency},  the competitive ratio achieved in the case where the predictor is perfect, and \textit{robustness}, the competitive ratio in the case where the predictor is pessimal.  Unlike the Hurwicz index of pessimism, in \cite{purohit2018improving,gollapudi2019online,angelopoulos2020online,wei2020optimal} consistency/robustness are \textbf{not} aggregated into a single performance measure. Instead, one computes the Pareto frontier of the two measures, and/or optimizes consistency for a desired degree of robustness. 

\textbf{We can employ a similar approach to define equilibria with predictions.} Specifically, we will keep from Nash equilibria (that count competitive ratio equilibria as a special case) the crucial assumption that \textit{agents are able to perfectly anticipate the aggregate effect of the actions of the other players}. However, instead of assuming (as in competitive ratio equilibria) that agents lack a predictor for their \textbf{own} future behavior, and optimize by minimizing their competitive ratio, we will specify acceptable degrees of robustness/consistency for the strategies employed by each player: 

\begin{definition} Suppose $\lambda_1,\ldots, \lambda_n, \mu_1,\ldots, \mu_n\in [1,\infty)\cup \{\infty\}$. 
A strategy profile $W=(W_1,W_2,\ldots, W_n)$ is an \emph{equilibrium with predictions with parameters $(\lambda_1,\ldots, \lambda_n, \mu_1,\ldots, \mu_n)$} iff: 
\begin{itemize}[leftmargin=*]
\item[(i). ] For every set of individual predictions $\widehat{T_1},\ldots, \widehat{T_n}$ and all $1\leq i\leq n$,  following 
$W_i$ is $\lambda_i$-competitive on input $W$. 
 \item[(ii). ] Every $W_k$ is $\mu_k$-consistent. That is, for all actual active times $(T_1,T_2,\ldots, T_n)$ we have $c_{W_{k}}(W,\widehat{T_k}=T_k)\leq \mu_k.$ 
\end{itemize} 
We speak of a \textbf{strict equilibrium} when, additionally, the following condition holds: 
\begin{itemize} 
 \item[(iii). ] Every $W_i$ is Pareto optimal on input $W$ relative to their robustness/consistency pair.  
\end{itemize} 
\end{definition} 

\begin{example}
Predictionless equilibria in the previous section are examples of equilibria with predictions: if $W=(W_1,W_2,\ldots, W_n)$ is an equilibrium without predictions and $W_i$ is $\lambda_i$-competitive on input $W$ then $W$ is an equilibrium with predictions with parameters $(\lambda_1,\lambda_1,\lambda_2,\lambda_2,\ldots, \lambda_n,\lambda_n)$. Note that generally in equilibria with predictions $\mu_k< \lambda_k$, since a correct predictor helps.  
\label{one}
\end{example}

\begin{example} Another example of an equilibrium with predictions is the following: agents get a self-predictor, which they follow blindly.  As noted in \cite{anon2025}, such an algorithm is 1-consistent but its robustness is $\infty$. When every agent uses this algorithm we get an equilibrium with predictions with parameters $(\infty,1,\infty,1, \ldots, \infty,1)$ 
\label{two}
\end{example} 

 \section{Equilibria with (Self) Predictions: Properties} 
\label{sec:ec:predictions}

It is possible in principle to use characterizations of approximate best-responses in \cite{anon2025} to completely describe equilibria with predictions for simple rational programs. The resulting conditions are, however, somewhat cumbersome. Instead, we show the following result, that shows that equilibria with predictions have properties reminiscent of their predictionless counterparts, characterized in Theorem~\ref{thm:ptb}:  

\begin{theorem} 
(a). In every program equilibrium with predictions with parameters $(\lambda_1,\ldots, \lambda_n, \mu_1,\ldots, \mu_n)$ such that $min(\lambda_i)<\infty$ the license gets eventually bought irrespective of the predictions $\widehat{T_1},\ldots, \widehat{T_n}$ the $n$ agents receive. 
That is, there exists $T_0$ such that in any run such that  $T_i>T_0$ holds for at least an agent $i$, the license gets bought no later than day $T_0$. 

(b). In every run of a program equilibrium with predictions with parameters $(\lambda_1,\ldots, \lambda_n, \mu_1,\ldots, \mu_n)$ such that each $W_i$ is a simple, rational online algorithm with others' predictions the following is true: on some day $r$ agents in some subset $S\subseteq [n]$ pledge amounts $(w_i)_{i\in S}$, $w_i\geq 1$, $\sum_{i\in S} w_i\geq B$, such that, for all $k\in S$ 
\[
\frac{w_k-1}{r}\leq \lambda_{k} \frac{B-1}{M_{*}(p)}+\lambda_k+\frac{1}{\lambda_k}-2\mbox{, if }r\leq M_{*}(p), 
\]
\[
r+w_k-1\leq \lambda_k (2B-1)+(\lambda_k+\frac{1}{\lambda_k}-2) \cdot B\mbox{, if }r>M_{*}(p).
\]
 Agents never pledge amounts $>0$ before day $r$. 

(c). The competitive ratio of agent $k\in S$ on the run described at point (b). is $1+min(\frac{B-1}{M_{*}(p)},\frac{w_k-1}{r})$ if $r\leq M_{*}(p)$, and $\frac{min(M_{*}(p)+B-1, r+w_k-1)}{M_{*}(p)}\mbox{ if }r> M_{*}(p)
$. 
\label{thm:ptb2} 
\end{theorem} 
\begin{proof} 
Reminiscent of that of Theorem~\ref{thm:ptb}, see the Appendix for details. 
\end{proof}

\section{Using algorithms with predictions in competitive ratio equilibria}
\label{sec:5}

In this section we investigate, first theoretically, then experimentally the following scenario: suppose $n$ agents interact by playing one competitive online equilibrium \textit{without predictions}, e.g. one of the equilibria with no self-predictions characterized in Theorem~\ref{thm:ptb}. 

Suppose now that one of the agent deviates from this equilibrium behavior by employing a self-predictor. Crucially, the other agent \textit{don't know that the agent deviates}, and continue playing their equilibrium strategies. How much can the deviating agent improve its competitive ratio by involving the predictor, in the case that the predictor is correct, without compromising too much its competitive ratio when the predictor is pessimal? 

Results from \cite{anon2025} provide a natural setting to test this scenario: an algorithm was developed with a tunable parameter $\lambda\in [0,1]$ that "interpolates" between the optimal algorithm with no self-prediction and the algorithm that blindly follows the predictor. The algorithm was inspired by the single-agent ski-rental with prediction problem studied in \cite{anon2025}, and is presented here as Algorithm~\ref{zeroprime-alg}\footnote{Compared to \cite{anon2025}, we have simplified the presentation of the algorithm by omitting from its specification the behavior for freeriding agents.}. 

\begin{algorithm}
   
    \textbf{Input: $(p,T)$}\\
       \textbf{Predictor: $\widehat{T}\geq 1$, an integer.}\\
       \vspace{-5mm}
       \begin{algorithmic}
    \State {choose $r_2=r_{2}(p,\lambda)$ to be the first day $\geq \lceil (1-\lambda) (r_0-1)+ \lambda r_1\rceil$ such that $P_{r_2}- \lambda OPT_{r_2}=min(P_{t}-\lambda OPT_{t}:t\geq \lceil (1-\lambda) (r_0-1)+ \lambda r_1\rceil ) $;\\
     if $P_{r_2}>P_{r_1}$ let $r_2=r_1$}
    \State {choose $r_{3}=r_{3}(\lambda,p)$ so that $\frac{P_{r_{3}}}{OPT_{r_{3}}(p)}\leq \lambda-1+\frac{1}{\lambda}c_{OPT}(p)$  minimizing the ratio $\frac{P_{r_{3}}}{r_{3}}$}
        \If{$(\widehat{T}\geq M_{\widehat{T}}(p))\;$}
      \State{ 
          buy on day $r_2$ (possibly not buying at all if $r_2 >T$)}
       \Else
       \State{
        buy on day $r_3$ (possibly not buying at all if $r_3>T$)}
      \EndIf
      \end{algorithmic} 
    \caption{An algorithm displaying a robustness-consistency tradeoff.}
    \label{zeroprime-alg}
\end{algorithm}

The case $\lambda= 1$ corresponds to the best-response play in equilibria without self-predictions. So one could implement equilibria in Theorem~\ref{thm:ptb} by having all non-freeriding agents $i$ play Algorithm~\ref{zeroprime-alg} with $\lambda_{i} = 1$ (freeriding agents already plays the optimal strategy). Assume, therefore, that \emph{some agent $k\in S$ actually does take into account the predictor, using Algorithm~\ref{zeroprime-alg} with parameter $\lambda_k\in (0,1)$}, instead of simply best-responding with $\lambda_k=1$. Crucially, assume that \textbf{this behavior is, however, invisible to all other agents, who still believe that $\lambda_k=1$.}

\subsection{Theoretical results}

The following result gives bounds on the consistency/robustness of Algorithm~\ref{zeroprime-alg} when playing an equilibrium: 

\begin{theorem} Assume $k\in S$. Then, using Algorithm~\ref{zeroprime-alg} with  $\lambda_k\in (0,1]$ in the competitive ratio equilibria from Theorem~\ref{thm:ptb} is $\frac{1}{\lambda_k}$-robust and $\beta_{k,\lambda}$-consistent, where $\beta_{k,\lambda}$ is given in Table~\ref{sample-table}. 
\label{eq-one-prediction} 
\end{theorem} 
\begin{table*}[ht]
  \caption{Values of consistency parameter $\beta_{k,\lambda}$ from Th.~\ref{eq-one-prediction}:}
  \label{sample-table}
  \centering
  \begin{tabular}{lll}
    \toprule
    Case   & Conditions     & Consistency $\beta_{k,\lambda}$ \\
    \midrule
    $r=1$ & $w_k=1\mbox{ or }w_k\geq 2,  \lambda \leq \frac{w_k(w_k-1)}{r(B-1)}$  & 1     \\
    $r=1$ & $w_k\geq 2, \lambda > \frac{w_k(w_k-1)}{r(B-1)}.$ & $w_k$  \\
    $2\leq r \leq M_{*}(p)$ & $\frac{w_k-1}{r\lambda}+\frac{(\lambda-1)^2}{\lambda} \leq 1-\frac{1}{B},
\frac{B+ \lceil \lambda r\rceil -1}{r+w_k-1}<1$
 & $1+max(\frac{\lceil \lambda r\rceil -1}{B},\frac{w_k-1}{r})$\\
 $2\leq r \leq M_{*}(p)$ & $\frac{B+ \lceil \lambda r\rceil -1}{r+w_k-1}<1, \frac{w_k-1}{r\lambda}+\frac{(\lambda-1)^2}{\lambda}> 1-\frac{1}{B}$ & $1+\frac{\lceil \lambda r\rceil -1}{B}$\\
 $2\leq r \leq M_{*}(p)$ & $\frac{w_k-1}{r\lambda}+\frac{(\lambda-1)^2}{\lambda}\leq 1-\frac{1}{M_{*}(p)}, 
 r+w_k-1\leq B$ & $1+\frac{w_k -1}{r}$\\
 $2\leq r \leq M_{*}(p)$  & $M_{*}(p)[\frac{w_k-1}{r\lambda}+\frac{(\lambda-1)^2}{\lambda}]> B-1,
r+w_k-1\leq B$ & 1\\
    $M_{*}(p)<r\leq 2B-1$    & $B+ \lceil \lambda r\rceil -1 < r+w_k-1$ &  $1+\frac{w_k -1}{r}$      \\
    $M_{*}(p)<r\leq 2B-1$ & otherwise  &  $\frac{r+w_k-1}{B}$        \\
    \bottomrule
  \end{tabular}
\end{table*}

\begin{observation} 
For $\lambda \leq \frac{M_{*}(p)(w_k-1)}{r(B-1)}$ one can easily see that the consistency bound $\beta_{k,\lambda}$ improves on the robustness guarantee for algorithms without predictions (i.e. competitive ratio) $1+\frac{w_k-1}{r}$.  So an agent $k\in S$ using a parameter $\lambda_k<  \frac{M_{*}(p)(w_k-1)}{r(B-1)}$ is able to improve upon equilibrium behavior if its predictor is perfect. 
\end{observation}

\begin{figure*}[ht]
\centering
\begin{tabular}{|c|c|c|c|}
\hline
\subf{\includegraphics[width=.20\textwidth]
{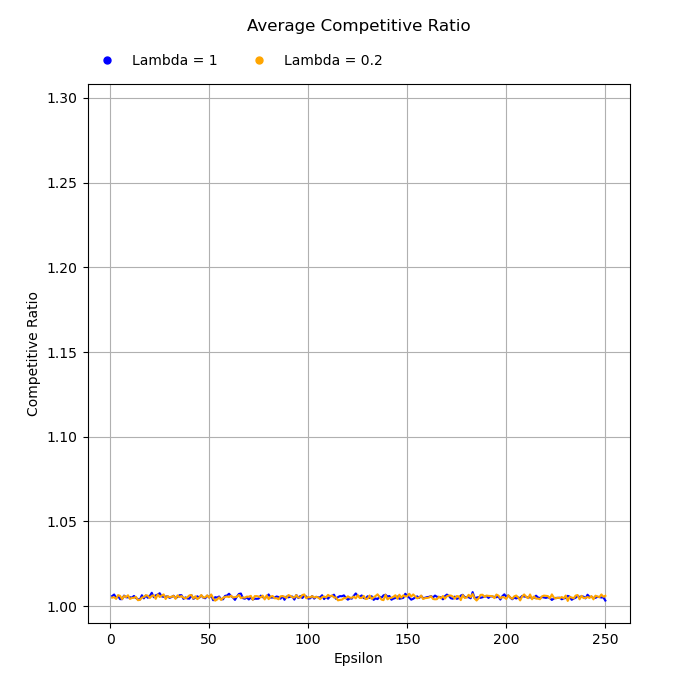}}{} 
&
\subf{\includegraphics[width=.20\textwidth]{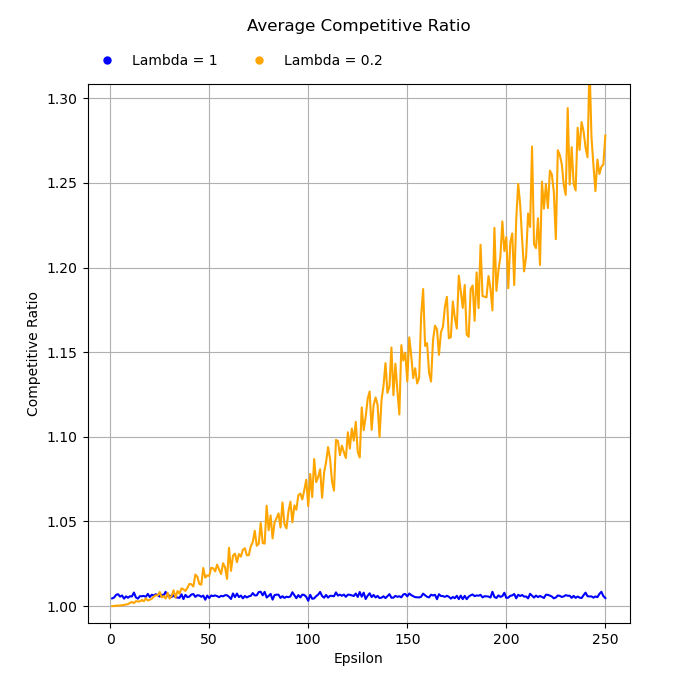}}{}
& 
\subf{\includegraphics[width=.20\textwidth]{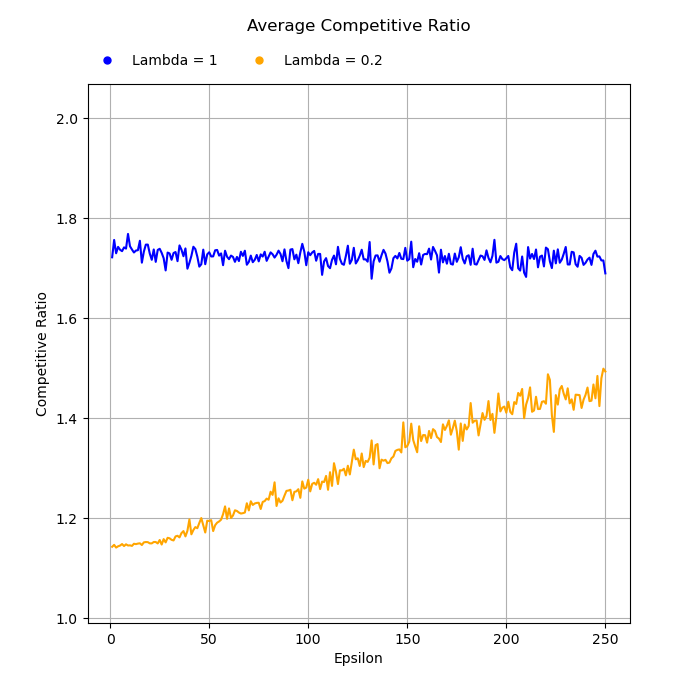}}{}
& 
\subf{\includegraphics[width=.20\textwidth]{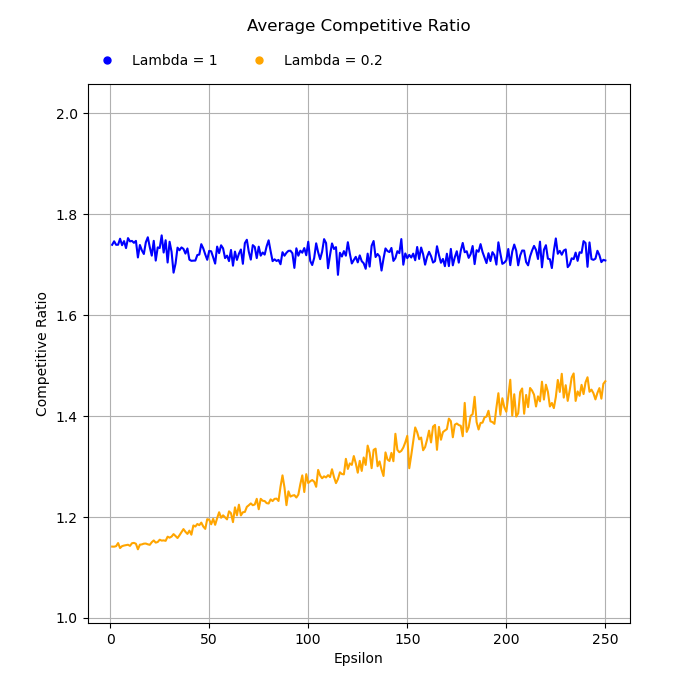}}{}
\\
\hline 
\subf{\includegraphics[width=.20\textwidth]
{Figure_7_equilibrium}
}{}
&
\subf{\includegraphics[width=.20\textwidth]{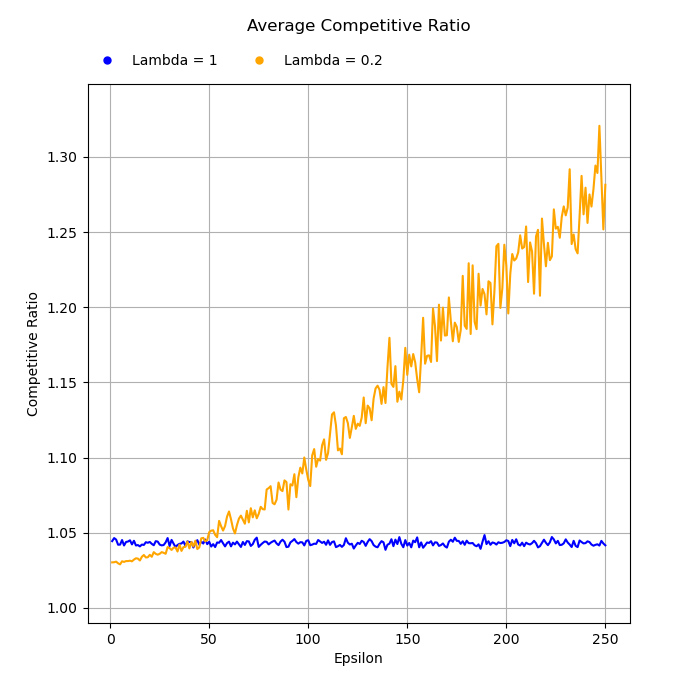}}{}
& 
\subf{\includegraphics[width=.20\textwidth]{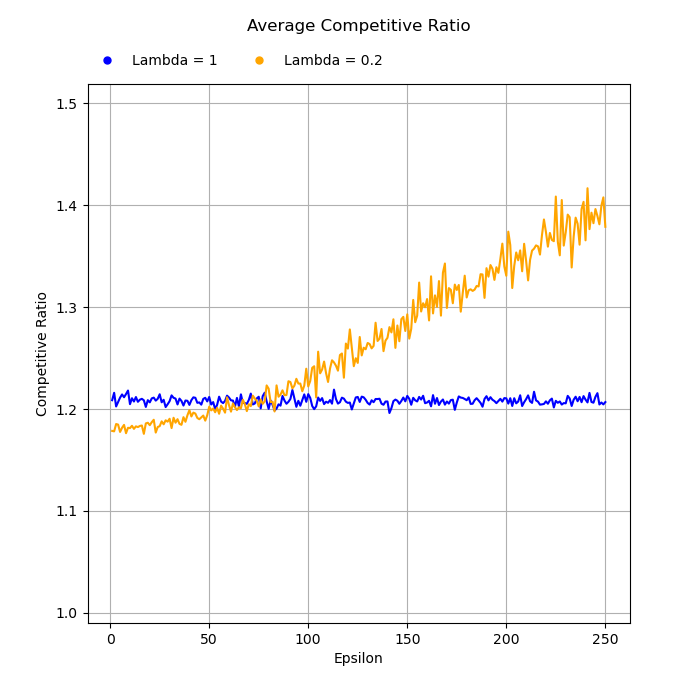}
}
     {}
& 
\subf{\includegraphics[width=.20\textwidth]{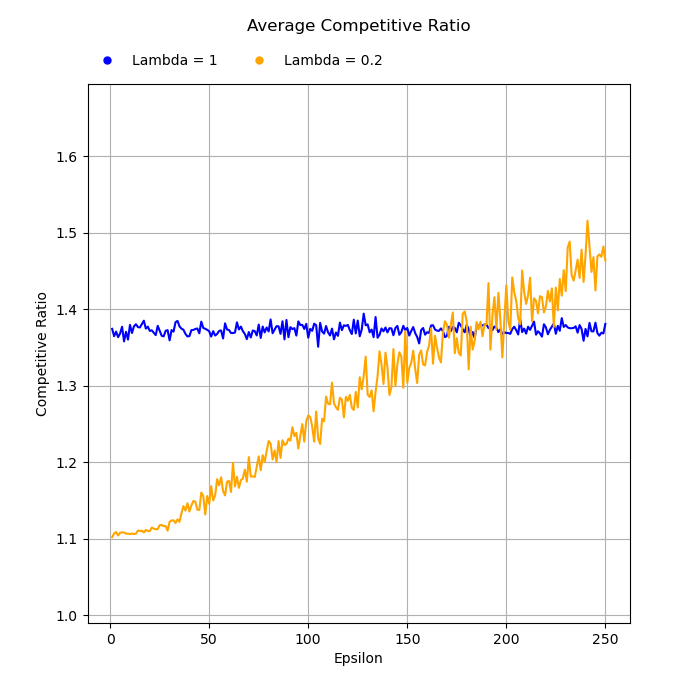}
}
     {}
\\
\hline
\end{tabular}
\label{fig-1}
\caption{Average consistency for ski-rental with predictions on  various equilibria  (a). "Trivial" equilibrium $(r=75,w=19)$. (b). Increasing $w$ by 1 makes equilibrium nontrivial: (r=75,w=20) (c,d): "Late" equilibria: (c) $(r=125, w=75)$, (d) $(r=125,w=15)$. (e-). The effect of increasing the group license acquisition date in "nontrivial" equilibria. $r=75$ (e). $w=20$ (same as (b)). (f) $w=30$, (g). $w=50$. (h). $w=70$.}
\end{figure*}

\begin{proofidea}

From the perspective of agent $k$ the problem is a ski-rental problem with costs $p_r=w_k$, $p_j=B$ for $j\neq r$. Therefore $P_r=r-1+w_k$, 
$P_j=j-1+B$ for $j\neq r$.

We use these specific prices as an input $p$ to Algorithm~\ref{zeroprime-alg}, 
and the following result from \cite{anon2025} (Theorem 6.4 in that paper), which gives a mathematical formula for the consistency of Algorithm~\ref{zeroprime-alg} on a given input $p$, given active time $T$ and prediction $\widehat{T}$: 

\begin{proposition} The competitive ratio of  algorithm~\ref{zeroprime-alg} on input $(p,T)$, prediction $\widehat{T}$,  is 
\begin{equation}
c_{A}(p,T;\widehat{T})\leq \lambda - 1+\frac{1}{\lambda}c_{OPT}(p).
\end{equation} 
On the other hand, if $T=\widehat{T}$ then 
\begin{equation}
 c_{A}(p,T;T)\leq max(\frac{P_{r_2}}{M_{*}(p)}, \frac{P_{r_{3}}}{r_{3}})\mbox{ if }r_3\leq M_{*}(p), 
\end{equation} 
\begin{equation}
 c_{A}(p,T;T)\leq \frac{P_{r_2}}{M_{*}(p)},\mbox{ otherwise. }
\end{equation}
\label{alg-predictions} 
\end{proposition} 

To apply these formula to computing the robustness/consistency bounds for Algorithm~\ref{zeroprime-alg}
we need, of course, to compute the quantities $c_{OPT}(p)$,  $M_{*}(p),r_2,r_3$ for our specific input $p$. 

This is accomplished via a case-by-case analysis, presented in detail in the Appendix. The analysis deals with the following cases: 
\begin{itemize}[leftmargin=*] 
\item[-] $r=1$, with subcases $w_k=1$ and $w_k\geq 2$ 
\item[-] $2\leq r\leq M_{*}(p)$ 
\item[-] $M_{*}(p)<r\leq 2B-1$ 
\end{itemize} 

The interesting quantity is (in several cases) $r_3$, since to compute the consistency guarantee we need to compare it to parameter $M_{*}(p)$.

\end{proofidea} 

\subsection{Experimental Results}

In \cite{anon2025} we benchmarked Algorithm~\ref{zeroprime-alg} under a model that assumed randomly generating total pledges of the other players, and investigated its average consistency/robustness. It goes without saying that this scenario does \emph{not} accurately represent equilibrium behavior: agents behave strategically, and there is no reason to suspect that their average pledges are, somehow, "random". 

In this section we subject Algorithm~\ref{zeroprime-alg} to an experimental investigation, assuming a setting similar to the one considered in Theorem~\ref{eq-one-prediction}.  Namely,  assume that agents play one of the equilibria described in Theorem~\ref{thm:ptb} and investigate the performance of an agent that deviates from equilibrium behavior by employing Algorithm~\ref{zeroprime-alg} with a value of $\lambda$ different from 1. The python code employed is the same as the one used for experiments in  \cite{anon2025}, and is already posted on github. A pointer to the code will be provided once the paper is published. 

Note that, although we benchmark the algorithm on several equilibria, \textbf{it makes no real sense to compare the performance of the algorithm between different equilibria}: the competitive ratio of the algorithm without predictions depends (Theorem~\ref{thm:ptb}) on the parameters $(w,r)$ of this equilibrium. For this reason every experiment will be performed (with several $\lambda$) on data arising from a single equilibrium.  As with the experiments in \cite{anon2025}, we aim to match the experimental setup of \cite{purohit2018improving} as well as we can. We will assume, therefore, that $B=100$, that $T\in [1,4B]$ is randomly chosen, that $\widehat{T}=T+\epsilon$, with $\epsilon$ being normally distributed with average 0 and standard deviation $\sigma$. We will choose $\lambda\in \{0.2,1\}$. Prices will correspond to an equilibrium: that is $P_{r}=r-1+w_k$ for some $r\leq 2B-1$, $P_{j}=j-1+B$ for $j\neq r$, with $w_k$ satisfying the appropriate condition in the characterization of equilibria (Theorem~\ref{thm:ptb}).

\begin{observation} When choosing a pair $(r,w)$ we need to check that it satisfies the conditions of Theorem~\ref{thm:ptb}, so that it is indeed an equilibrium. Here is a sample computation for 
the pair $(r=75, w=70)$: $r(B-1)=75*99=7425$. $(w-1)\cdot min(B,r+w-1)=69 \cdot min(100,75+70-1)=6900$. So $(w-1)\cdot min(B,r+w-1)<r(B-1)$. In this case,  the worst-case robustness theoretical guarantee is $1+\frac{w-1}{r}=1+\frac{69}{75}=1.92$. Note that the results in Figure~\ref{fig-1} (h) are (well) below this guarantee, since what we are measuring/plotting is average (rather than worst-case) competitiveness. 
\end{observation}



\subsection{"Trivial" equilibria} 

It turns out that for some combinations of values of the parameters $(r,w)$ something remarkable happens: the resulting equilibria are "trivial". An example is displayed in  Figure~\ref{fig-1} (a): $r=75$, $w_k=19$. This is an equilibrium: $r-w_k+1=93<100=B$. Also, the consistency worst-case upper bound for $\lambda=1$ is $1+\frac{18}{75}=1.24$. 

However, in practice the algorithm is almost 1-consistent on the average. This seems true both for $\lambda=1$ and $\lambda = 0.2$,  and seems independent of the error (see Figure~\ref{fig-1} (a).)  An explanation is that there seem to be very few cases where the performance of the algorithm is not 1. Even when $\sigma=250$, out of 1000 samples there were only 36 such samples for $\lambda=1$, and only 27 samples for $\lambda=0.2$. To these samples the upper bound of $1.24$ applies, but the actual competitive ratios (for values of $T$ which are random, not worst case) are even smaller. So it is not \textit{that} surprising that the average competitive ratio is very close to 1. Further decreasing $w$ (e.g. $r=75,w=18,17,16,1$), seems to produce similarly "trivial" equilibria (not pictured).

\subsection{"Nontrivial" equilibria}

If, instead, we increase $w$ the picture changes: Even increasing $w$ from 19 to $20$, that is by the smallest possible amount, is enough to make the algorithm with predictions improve (albeit slightly) on the predictionless algorithm for tiny values of the variance $\sigma$, while performing worse for larger variances. Further increasing $w$ seems to have a dual effect: it degrades the performance of the predictionless algorithm, on one hand (not surprising, given the worst-case guarantee in Theorem~\ref{thm:ptb}, which also degrades). On the other hand the algorithm with predictions starts outperforming the predictionless algorithm for a larger range of values of the standard deviation $\sigma$. 
We display experiments supporting these claims this by choosing $r=75$ and varying the values of $w$ to $w=20,30,50,70$ (Figure~\ref{fig-1} (d-f)). The performance of the algorithm is similar for all these "nontrivial" equilibria: better than that of the predictionless algorithm for small values of $\sigma$, and degrading as $\sigma$ grows. 

\subsection{"Late equilibria"} 

In Theorem~\ref{thm:ptb} the conditions on the pledged amounts at equilibrium differ between the cases $r\leq B$ and $B+1\leq r\leq 2B-1$. So far, all equilibria tested experimentally were of the first type. It is natural to inquire whether switching to equilibria of the second type (we will call them "late equilibria") makes any difference. 

In Figure~\ref{fig-1} (c,d) we present plots for two such "late equilibria", $(r=125, w=75)$ and $(r=125,w=15)$. The answer to the previous question seems to be negative: the qualitative behavior of the algorithm seems similar to that displayed in nontrivial "early" equilibria. If anything, the algorithm with predictions seems to become even more robust/useful for such "late" equilibria. 

\section{Limitations of our Study} 

As with results in \cite{anon2025}, our results only deal with \emph{pure} equilibria. It is an open problem to adapt the results from that paper (and this one) to the case of mixed equilibria. 

Second, our results are currently limited by our present inability to provide a convincing measure of the accuracy of the others' prediction, interpolating between perfect and pessimal predictions. This limits our definition of equilibria with predictions to self-predictions. Giving a definition of equilibria with predictions that deals with others' predictions as well is an interesting challenge.  
 
\section{Conclusions and Open Problems} 

Our work raises many open questions. 
The most important of them are, we believe, those discussed in the previous section. Equally important is the problem of obtaining more general results on equilibria with predictions, perhaps in the general framework sketched in \cite{anon2025}. 

On the other hand predictions in this paper are one-shot events: agents get this knowledge from the outset. A more natural model is that of an algorithm with \emph{continuous, time-limited predictions}\footnote{real-life predictions, such as weather forecasts, have this nature.}. Perhaps a study of (a multiagent version of) \textit{the Bahncard problem} \cite{fleischer2001bahncard} in a model with an evolving prediction window is feasible. Note that the single-agent version of this problem has already been considered in a perspective with predictions \cite{ding2005risk}, but we believe that even this single-agent case deserves more attention (and work). 

There exists a huge, multi-discipline literature that deals with learning game-theoretic equilibria. In such a setting, an equilibrium arises through repeated interactions (see, e.g. \cite{kalai1993rational,young2004strategic,JMLR:v24:22-0131}). Applying such results to the setting with predictions and, generally, allowing predictions to be \emph{adaptive} (i.e. depend on previous interactions) is left to further research. 

Also interesting is the issue of obtaining tradeoffs between parameters of approximate equilibria (similar to \cite{gollapudi2019online,angelopoulos2020online,wei2020optimal} for the one-agent case). This amounts, in the setting of Theorem~\ref{thm:ptb2}, to deciding for which  $(\lambda_1,\ldots, \lambda_n, \mu_1,\ldots, \mu_n)$ do such equilibria exist. 

Finally, our notion of equilibrium with predictions builds on 
Nash equilibria. But there exist other game-theoretic alternatives: A \textit{Kantian competitive ratio equilibrium} (see e.g. \cite{roemer2019we,istrategame} for details about Kantian optimization) is a strategy profile $\Delta_{P}=(P,P,\ldots, P)$ that minimizes the competitive ratio of $P$ among all inputs of type $\Delta_Q$.  A model studied in \cite{engelberg2016equilibria} characterized (without using this terminology) Kantian equilibria (without predictions) for a number of problems. It would be interesting to understand predictionless Kantian equilibria for multiagent ski-rental. It would be even \textbf{more} interesting to define a suitable notion of Kantian equilibria with predictions. 


 
\bibliography{/Users/gistrate/Dropbox/texmf/bibtex/bib/bibtheory}

\section{Appendix}

First we present full proofs of the theoretical results in the paper. Note that \textbf{we have made no effort to write completely new text, when presenting the Proof Highlights in the main paper}. In other words, the shorter text given there reuses liberally phrases from these full proofs. 

For all notations not defined in the main text we refer the reader to \cite{anon2025}. 

\subsection{Proof of Theorem~\ref{thm:predictionless}}

\begin{proof} In the proof we will need the following characterization of best-response in the Ski-Rental Problem with Varying Prices (Theorem 4.4. in \cite{anon2025}, see that paper for a proof): 

\begin{proposition} Optimal competitive algorithms in Ski Rental with Known Varying Prices are: 
\begin{itemize}[leftmargin=*]
\item[(a).] if a free day exists on day $d=M_{*}(p)+1$ or a bargain day exists on day $d=M_{*}(p)$, then the agent that waits for the first such day, renting until then is 1-competitive. The only other 1-competitive algorithm exists when the first bargain day $b=M_{*}(p)$ is followed by the first free day $f=M_{*}(p)+1$; the algorithm rents on all days, waiting for day $f$.  
\item[(b).] Suppose case (a) does not apply, i.e. the daily prices are $p_1,p_2,\ldots, p_n \ldots $, $1< p_i\leq B$ for $1\leq i\leq M_{*}(p)$, $p_{M_{*}(p)+1}\geq 1$. Then no deterministic algorithm can be better than $c_{OPT}(p)$-competitive, where 
$ c_{OPT}(p):=
min(\{\frac{P_r}{r}:r\leq M_{*}(p)\}\cup\{ \frac{Q_{M_{*}(p)}}{M_{*}(p)}\})$. Algorithms that buy on a day in $argmin(\{\frac{P_r}{OPT_r}:1\leq r\leq M_{*}(p)\}\cup \{\frac{P_r}{M_{*}(p)}: r\geq M_{*}(p), P_r=Q_{M_{*}(p)}\})$ are the only ones to achieve this ratio. 
\end{itemize} 
\label{thm-full-free}
\end{proposition} 

We use this result as follows: We take the single agent perspective in multi-agent ski rental, and investigate the best response to the other agents' behavior. Fixing the algorithms of the other agents means that the agent knows \emph{how much money it would need to contribute on any given day to make the group license feasible}. This effectively reduces the problem to computing a best response to the ski-rental problem with known varying prices, as described in Definition~\ref{ski-rental-known-varying}.
\begin{itemize}[leftmargin=*] 
\item[(a).] Assume otherwise: suppose that $(W_1, W_2,\ldots W_n)$ was a program equilibrium where the license never gets bought. This is because pledges on any day never add up to the licensing cost $B$. Consider the perspective of agent 1. It will never have a free day, hence its best response is to pledge the required amount on its optimal day/first bargain day (if it exists), at the needed cost.
\item[(b).]  

\begin{itemize}[leftmargin=*] 

\item[(i).] To have a competitive ratio equilibrium, every agent $i=1,\ldots, n$ has to play a best-response in the ski-rental problem with varying prices $p^i=(p_j)_{j\geq 1}$,  $p_j=max(B-\sum_{k\neq i}f_{k}(j),0)$. Note that, strictly speaking, the agent is able to make inconsequential pledges on days $j\neq r$: this departs from the parallel with ski-rental with varying prices, but does not affect agents' competitive ratio anyway, since the license doesn't get bought. By point (a), the license is bought on some day $r$. So $\sum_{j=1}^{n} f_{j}(r)\geq B$. In fact $\sum_{j=1}^{n} f_{j}(r)=B$: if it were not the case then any participating agent could lower its price paid on day $r$ (hence its competitive ratio) by pledging just enough to buy the license. Finally, note that (by definition) for every $i=1,\ldots n$, $Z_i=M_{*}(p^i)$. 

\item[(ii.)] Consider  such an agent $i$. The hypothesis implies the fact that it will have a free day on day $r\leq Z_i+1$. By Proposition~\ref{thm-full-free}(a), it is optimal (1-competitive) for the agent $i$ not to pledge any positive amounts and simply wait for its free day.
Therefore $f_i(r)=0$. 
In fact, the agent can make inconsequential pledges before day $r$, as they don't lead to acquiring the license.

\item[(iii).] Consider such an agent $i$. The hypothesis implies the fact that it has a bargain day on a day $k\leq min(r,Z_i)$. Given that the license is bought on day $r$, it is either the case that the first such day is $k=r$ or that $k=r-1$ and $r$ is a free day for agent $i$. In this case the agent doesn't have to pledge on day $r-1$ since it gets the license for free on day $r$.

\item[(iv).] Consider now an agent $i$ such that $f_i(r)> 1$. From the agent's perspective $P_j=j-1+B-\sum_{k\neq i} f_{k}(j)$. In particular $M_{*}(p^i)\leq P_{1}\leq B$. 
The agent has no free/bargain days. We have to show that day $r$ is an optimal day to buy for agent $r$.

Indeed, first consider $j\neq r$, $j\leq Z_i$. 
By the assumption  we have $\sum_{k\neq i} f_{k}(j)\leq B-1-\frac{j}{OPT_r}(f_i(r)-r-1)$ and $j=OPT_j$, 
hence $P_j=j-1+B-\sum_{k\neq i} f_{k}(j)\geq j-1+B- (B-1)+\frac{j}{OPT_r}(f_i(r)+r-1)= \frac{OPT_j}{OPT_r} \cdot (r+f_i(r)-1)= \frac{OPT_j}{OPT_r} \cdot P_{r}$ for $j< r$. 


Similarly, one can show that if $j>Z_i=M_{*}(p)$ then $\frac{P_r}{OPT_r}\leq \frac{P_j}{M_{*}(p)}=\frac{P_j}{OPT_j}$. 

Hence it is optimal for the agent to pledge on day $r$, and the agent does so. 
\end{itemize} 
\end{itemize} 
\end{proof}

\subsection{Proof of Theorem~\ref{foo}}

\begin{itemize} 
\item[-] if $i$ is an agent in case (ii) then it pays $r$ dollars, and the optimum is still $r$.  
\item[-] if $i$ is an agent in case (iii) then it either pays $r$ dollars and the optimum is $r$ (first subcase), or it pays $r-1$ dollars (renting for the first $r-1$ days and receiving the license for free on day $r$), and the optimum is $r-1$ (second subcase). 
\item[-] if $i$ is an agent in case (iv) then it pays $f_{i}(r)+r-1$ dollars. $OPT_{r}$ is $min(r,Z_i)$. 
\end{itemize} 

\subsection{Proof of Theorem~\ref{thm:ptb}}

\begin{itemize}[leftmargin=*] 

\item[(a).]  First of all, it is easy to see that the strategy profiles mentioned in the result are Nash equilibria: to prove this we have to simply show that every agent uses the optimal competitive algorithm, given the behavior of other agents. 

Consider first an agent $i\not\in S$. It will have a free day on day $r$. Given that (from the agent's perspective) $p_j=B$ for $j\neq r$,  $p_r=0$, so we have $M_{*}(p)=min(j-1+B: j\neq r, r-1+p_r)= min(B,r-1)$. If $r\leq B+1$ then it is optimal (1-competitive) for the agent $i$ not to pledge any positive amounts and simply wait for its free day, which realizes minimal cost $M_{*}(p)$ (and the agent does so). If $r>B=M_{*}(p)$ then to find the optimal action we have to compare ratios $\frac{P_j}{j}$ for $j\leq M_{*}(p)$, together with the cost $\frac{Q_{M_{*}(p)}}{M_{*}(p)}$. Given the concrete form of sequence $(p_j)$ we have $\frac{P_j}{j}=1+\frac{B-1}{j}$ for $j\leq B$, the minimal ratio being obtained for $j=B$. On the other hand $P_s=s-1+B\geq 2B-1$ if $s\geq B, s\neq r$, while $P_r=r-1$. As long as $r\leq 2B$, the minimal cost is attained at $r$, while 
$\frac{P_r}{OPT_r}=\frac{r-1}{B}\leq \frac{2B-1}{B}$. So if $r\leq 2B$ then it is optimal to wait for the free day $r$, and the agent does so.  

Consider now an agent $i\in S$.  From the agent's perspective $p_j=B$ for $j\neq r$, $p_r=w_i$, so we have $M_{*}(p)=min(j-1+B: j\neq r, r-1+w_i)= min(B,r-1+w_i)$. Also, note that none of the days different from day $r$ is a free/bargain day, since $B>1$. 

If $w_i= 1$ and $r\leq B$ then $M_{*}(p)=min(B,r)=r$. 
So the agent agent has a bargain day on day $r\leq M_{*}(p)$. Furthermore, day $r+1$ is \textbf{not} a free day ($p_{r+1}=B$) so it is optimal for the agent to buy on day $r$, and the agent does so. 

If, on
 the other hand $w_i>1$ or $w_i=1$ and $r>B$ then we have to show that day $r$ is an optimal day to buy for agent $i$.

\begin{itemize}[leftmargin=*] 

\item{Case 1: $w_i>1$ and $r\leq M_{*}(p)$.} We have $\frac{P_r}{OPT_r}=1+\frac{w_i-1}{r}$ if $r\leq M_{*}(p)$. 

If $j<r\leq M_{*}(p)$ then $B-1\geq w_i-1$ and $\frac{P_j}{j}=
1+\frac{B-1}{j}$. The inequality $\frac{P_j}{j}\geq \frac{P_r}{r}$ follows from these two inequalities. 

If, on the other hand $r<j\leq M_{*}(p)\leq B$ then, since $w_i-1\leq r\frac{B-1}{ M_{*}(p)}$, we have $\frac{P_r}{r}\leq 1+\frac{B-1}{M_{*}(p)}\leq 1+\frac{B-1}{j}=\frac{P_j}{j}$. Finally, if  $j>M_{*}(p)\geq r$ then since $w_i-1\leq r\frac{B-1}{ M_{*}(p)}$, we have $\frac{P_r}{r}=1+\frac{w_i-1}{r}\leq 1+\frac{B-1}{M_{*}(p)}\leq \frac{j+B-1}{M_{*}(p)}=\frac{P_j}{OPT_j}$. Hence if $r\leq M_{*}(p)$ then it is optimal for the agent to pledge on day $r$, and the agent does so.

\item{Case 2: $w_i=1$ and $r>B$.}

In this case $M_{*}(p)=B$. To compute the optimal algorithm we have to compare $\frac{P_j}{j}$ for $j=1,\ldots, B$ with 
$\frac{Q_{M_{*}(p)}}{M_{*}(p)}$. 

On one hand, since $j\leq B<r$, $\frac{P_j}{j}=1+\frac{B-1}{j}\geq 1+\frac{B-1}{B}$. 

On the other hand $Q_{B}=min(2B-1,r)$. So $\frac{Q_{M_{*}(p)}}{M_{*}(p)}=min(1+\frac{B-1}{B},\frac{r}{B})$. For $r\leq 2B-1$ it is optimal for the agent to pledge on day $r$, and the agent does so.

\item{Case 3: $w_i>1$ and $r>B$.}

Indeed, $P_{j}=j-1+B$ for $j\neq r$, $P_{r}=r-1+w_i=r$. So for $j\neq r$ we have $\frac{P_j}{OPT_j}=1+\frac{B-1}{j}$ if $j\leq B$, $\frac{P_j}{OPT_j}=1+\frac{j-1}{B}$ if $B<j\neq r$.   On the other hand $\frac{P_r}{OPT_r}=\frac{r-1+1}{B}=\frac{r}{B}$. As long as $r\leq 2B-1$, $\frac{P_r}{OPT_r}\leq \frac{P_{j}}{OPT_j}$. 

\end{itemize}

Conversely, assume that $(W_1,W_2,\ldots, W_n)$ is a program equilibrium such that the license gets bought on day $r$. If the total pledges on day $r$ were larger than $B$ the one of the agents pledging on day $r$ could lower its pledge, therefore lowering its competitive ratio. Let $S$ be the set of agents whose pledges are positive on that day. By the previous discussion we have $\sum_{i\in S} w_i = B$. 

Assume now that some agent was pledging on a different day than $r$. Then either this pledge is inconsequential or it is too late. The first alternative cannot happen: each agent pledges the amount needed to complete the sum pledged by other agents up to $B$. So the pledge is late. In this case the agent has a free day by the time it would be optimal to pledge, so it would be  optimal to refrain from pledging. 

Consider now an agent $k\in S$ that is pledging on day $r$. There are two possibilities: 
\begin{itemize}[leftmargin=*] 
\item[-] Day $r$ is a bargain day for $k$. That is $w_k= 1$. 
\item[-] Day $r$ is an optimal day for agent $k$. Note that 
the prices the agent faces are $p_i=B$ on days $i\neq r$,  $p_{r}=w_{k}$. In order for pledging on day $r$ to be optimal (with competitive ratio equal to  $\frac{P_r}{r}$) we have to have 
\begin{itemize}[leftmargin=*] 
\item[-] $\frac{P_i}{i}\geq \frac{P_r}{OPT_r}$ for $i\leq M_{*}(p),i\neq r$. That is 
$1+\frac{B-1}{i}\geq 1+\frac{w_k-1}{r}$ if $r\leq M_{*}(p)$, 
$1+\frac{B-1}{i}\geq 1+\frac{w_k-1}{M_{*}(p)}$ if $r>M_{*}(p)$. Let's treat the two cases separately: 
\begin{itemize}[leftmargin=*] 
\item{$r\leq M_{*}(p)$:} The inequality is true if $i\leq r$ since $w_k\leq B$. To make this true for all days $i$, $r<i\leq M_{*}(p)$ we also have to have $\frac{B-1}{M_{*}(p)}\geq \frac{w_k-1}{r}$, that is $(w_k-1)M_{*}(p)\leq r(B-1)$. 
\item{$r>M_{*}(p)$:} The inequality is true since $w_k\leq B$. 
\end{itemize} 
\item[-] $\frac{P_i}{M_{*}(p)}\geq \frac{P_r}{OPT_r}$ for $i\geq M_{*}(p)$. That is $\frac{i-1+B}{M_{*}(p)}\geq \frac{r-1+w_k}{r}$ if $r\leq M_{*}(p)$, $P_i\geq P_r$ if $r>M_{*}(p)$. The first inequality requires that $\frac{B-1}{M_{*}(p)}\geq  \frac{w_k-1}{r}$ (since the bound for $i=M_{*}(p)$ is tightest)
i.e. $(w_k-1)M_{*}(p)\leq (B-1)r$. 

The second inequality implies that $P_{M_{*}(p)}=2B-1\geq P_{r}=r+w_k-1$, since $M_{*}(p)=B$. From this it follows that $r\leq 2B-1$. 
\end{itemize} 
  \end{itemize} 
 \item[b.] An easy application of Proposition~\ref{thm-full-free}. 
 \end{itemize}

\subsection{Proof of Theorem~\ref{thm:ptb2}}
\begin{proof} 
\begin{itemize} 
\item[(a).] Assume otherwise: suppose that $(W_1, W_2,\ldots W_n)$ was a program equilibrium where for some combination of predictions the license never gets bought. This is because pledges on any day never add up to the licensing cost $B$. Consider the perspective of agent $i$. It will never have a free/bargain day. To keep the robustness under control it will need to pledge the required amount on a day $r$ s.t.
\[
\frac{P_{r}}{OPT_{r}}\leq \frac{1}{\lambda_i}-1+\lambda_i \frac{P_{r_1}}{r_1}
\]
that is $r\leq P_{r}\leq T_0:= M_{*}(p)(\frac{1}{\lambda_i}-1)+ M_{*}(p)\lambda_i \frac{P_{r_1}}{r_1}$. 

\item[(b).]  Consider an arbitrary run $p$. Suppose $k$ is an agent that pledges amount $w_k$ to buy the group license on day $r$. From the perspective of agent $k$, the set of prices it experiences in the associated ski-rental problem with varying prices is as follows: $P_r=r-1+w_k$ if $r$ is the day the algorithm pledges on, $P_j=j-1+B$, otherwise. This is because the agents are rational -  they don't pledge unless their pledge leads to buying the license on that day. 

For such a sequence, $M_{*}(p)$ is either $B$ or $r-1+w_k$ and is reached either on day 1 or on day $r$. Indeed, if $r=1$ then $w_k\leq B$ so $M_{*}(p)=w_k=r-1+w_k=min(r-1+w_k,B)=w_k$. If $r>1$ then $M_{*}(p)=min(B,r-1+w_k)$, since $B$ is the price the agent would have to pay on day 1. Therefore 
\[
M_{*}(p)=\left\{\begin{array}{l}
B \mbox{ if } r-1+w_k>B\\
r-1+w_k\mbox{ otherwise.}
\end{array} 
\right.
\]
\[
i_{*}=\left\{\begin{array}{l}
1 \mbox{ if } M_{*}(p)=B\mbox{ or }r=1\\
r\mbox{ otherwise.}
\end{array} 
\right.
\]

To determine $r_1$ note that  $\frac{P_j}{j}=1+\frac{B-1}{j}$ for $j\neq r, j\leq M_{*}(p)$, while $\frac{P_{r}}{OPT_r}=1+\frac{w_k-1}{r}$ if $r\leq M_{*}(p)$, while $\frac{P_{r}}{OPT_r}=\frac{r+w_k-1}{M_{*}(p)}$ if $r> M_{*}(p)$. 
On the other hand $Q_{M_{*}(p)}=M_{*}(p)+B-1$ if $r\leq M_{*}(p)$, $Q_{M_{*}(p)}=min(M_{*}(p)+B-1,r+w_k-1)$, otherwise. So $\frac{Q_{M_{*}(p)}}{M_{*}(p)}=min(1+\frac{B-1}{M_{*}(p)},\frac{r+w_k-1}{M_{*}(p)})$ if $r>M_{*}(p)$, $1+\frac{B-1}{M_{*}(p)}$ otherwise. 

Note that the case $r=M_{*}(p), \frac{B-1}{M_{*}(p)-1}\leq \frac{w_k-1}{r}$ is impossible, since $w_k\leq B$. So
\begin{itemize}[leftmargin=*] 
\item If $r\leq M_{*}(p)$ then 
\[
r_1=\left\{\begin{array}{l}
M_{*}(p) \mbox{ if } M_{*}(p)\neq r, \frac{B-1}{M_{*}(p)}<\frac{w_k-1}{r}\\
r\mbox{ if } M_{*}(p)\neq r, \frac{B-1}{M_{*}(p)}\geq \frac{w_k-1}{r}\\
M_{*}(p)\mbox{ if } M_{*}(p)=r\\
\end{array} 
\right.
\]
\item If $r>M_{*}(p)$ then 
\[
r_1=\left\{\begin{array}{l}
M_{*}(p) \mbox{ if } M_{*}(p)\neq r, 1+\frac{B-1}{M_{*}(p)}<\frac{r+w_k-1}{M_{*}(p)}\\
r\mbox{ if } M_{*}(p)\neq r, 1+\frac{B-1}{M_{*}(p)}\geq \frac{r+w_k-1}{M_{*}(p)}\\
r\mbox{ if } M_{*}(p)=r\\
\end{array} 
\right.
\]

\end{itemize} 
Noting that $M_{*}(p)=r\Rightarrow  \frac{B-1}{M_{*}(p)}\geq \frac{w_k-1}{r}$, we infer that
\[
\frac{P_{r_1}}{r_1}=c_{opt}(p) = \left\{\begin{array}{l}
1+\frac{B-1}{M_{*}(p)} \mbox{ if } \frac{B-1}{M_{*}(p)}< \frac{w_k-1}{r}\\
1+\frac{w_k-1}{r}\mbox{ otherwise }\\
\end{array} 
\right.
\]
hence $
c_{opt}(p) = 1+min(\frac{B-1}{M_{*}(p)},\frac{w_k-1}{r})\mbox{ if }r\leq M_{*}(p)$, 
and 
\begin{align*}
\frac{P_{r_1}}{OPT_1}=c_{opt}(p) = \left\{\begin{array}{l}
1+\frac{B-1}{M_{*}(p)} \mbox{ if } B-1+M_{*}(p)< r+w_k-1\\
\frac{r+w_k-1}{M_{*}(p)}\mbox{ otherwise }\\
\end{array} 
\right.
\\
= \frac{min(M_{*}(p)+B-1, r+w_k-1)}{M_{*}(p)}\mbox{ if }r> M_{*}(p), 
\end{align*}
(this solves point (c)). 
We want to find conditions such that pledging on day $r$ is $\lambda_k$-robust. For $r\leq M_{*}(p)$ this is $\frac{w_k-1}{r}\leq$
\[
\leq \frac{1}{\lambda_k}+ \lambda_k c_{OPT}(p)-2= \lambda_k \cdot min(\frac{B-1}{M_{*}(p)},\frac{w_k-1}{r})+\lambda_k+\frac{1}{\lambda_k}-2
\]

Since $\lambda_k\geq 1$ and $w_k\geq 1$,  $\frac{w_k-1}{r}\leq \lambda_k \frac{w_k-1}{r}$. Also $\lambda_k+\frac{1}{\lambda_k}\geq 2$. 

So when $r\leq M_{*}(p)$ all we have to have is 
\[
\frac{w_k-1}{r}\leq \lambda_{k} \frac{B-1}{M_{*}(p)}+\lambda_k+\frac{1}{\lambda_k}-2,
\]
  the condition in the Theorem. 

Similarly, when $r>M_{*}(p)$ we need to have 
\[
\frac{r+w_k-1}{M_{*}(p)}\leq \lambda_k \frac{min(M_{*}(p)+B-1, r+w_k-1)}{M_{*}(p)}+\lambda_k+\frac{1}{\lambda_k}-2
\]
Again, one of the bounds is trivially true, so we need: 
\[
\frac{r+w_k-1}{M_{*}(p)}\leq \lambda_k \frac{2B-1}{M_{*}(p)}+\lambda_k+\frac{1}{\lambda_k}-2
\]
or 
\[
r+w_k-1\leq \lambda_k (2B-1)+(\lambda_k+\frac{1}{\lambda_k}-2) \cdot B
\]
\item[(c).] Similar to point (b). of Theorem~\ref{thm:ptb}. 
\end{itemize} 
\end{proof} 

\subsection{Proof of Theorem~\ref{eq-one-prediction}}

From the perspective of agent $k$ the problem is a ski-rental problem with the following costs: $p_r=w_k$, $p_j=B$ for $j\neq r$. Therefore $P_r=r-1+w_k$, 
$P_j=j-1+B$ for $j\neq r$.  Applying Algorithm 2 to this input sequence yields the following quantities: 
\begin{itemize}[leftmargin=*] 
\item[-] If $r=1$ then $M_{*}(p)=min(i-1+B:i\neq r, r-1+w_k)=min(B,w_k)=w_k$. Therefore  $i_{*}=r_0=1$. 
To compute $r_1$ we need to compare $\frac{P_1}{1}=w_k$,  $\frac{P_{i}}{i}=1+\frac{B-1}{i}$ for $i<w_k$ and $\frac{Q_{M_{*}(p)}}{M_{*}(p)}=\frac{P_{w_k}}{w_k}$. 
\begin{itemize}[leftmargin=*] 
\item If $w_k=1$ then $r_1=1$, so $r_2=1$. $c_{OPT}(p)=P_1=1$. To choose $r_{3}$ we need that $P_{r_{3}}\leq OPT_{r_{3}}(\lambda-1+\frac{1}{\lambda})$, minimizing the ratio $\frac{P_{r_{3}}}{r_{3}}$. For $j\geq 2$ $P_{j}=j-1+B$ so $\frac{P_{j}}{j}=1+\frac{B-1}{j}>1$. So $r_{3}=1$, since $\frac{P_1}{1}=1$. In this case the algorithm always buys on day 1 and is 1-competitive and 1-robust. 
\item If $w_k\geq 2$ then since $w_k\leq \frac{P_{w_k}}{w_k}=1+\frac{B-1}{w_k}$ 
(equivalently $\frac{w_k(w_k-1)}{B-1} \leq 1$) and $1+\frac{B-1}{w_k}<1+\frac{B-1}{i}$ for $i<w_k$, 
we have $r_1=1$. So $r_2=1$, $c_{OPT}(p)=w_k$. Also $\frac{P_{r_2}}{M_{*}(p)}=1$. 

The condition for $r_{3}$ reads $\frac{P_{r_{3}}}{OPT_{r_{3}}}\leq \lambda-1+\frac{1}{\lambda}w_k$. 
Since $\frac{P_{i}}{i}=1+\frac{B-1}{i}$, to find $r_3$
we attempt to get the largest possible value which satisfies this condition. 

First of all, $r_{3}=1$ works, since condition reads $\frac{w_k}{1}\leq \lambda-1+ \frac{w_k}{\lambda}$, or $(w_k-\lambda)(\frac{1}{\lambda}-1)\geq 0$ which is true, since $\lambda \in (0,1]$. 

If $2\leq r_{3}\leq M_{*}(p)$, $r_3\neq r$, then condition reads $\frac{P_{r_{3}}}{OPT_{r_{3}}}=1+\frac{B-1}{r_{3}}\leq \lambda-1+\frac{w_k}{\lambda}$. This is actually a lower bound on $r_{3}$: 
\[
M_{*}(p)=w_k\geq r_3\geq \frac{\lambda(B-1)}{\lambda^2-2\lambda+w_{k}} =\frac{\lambda(B-1)}{\lambda^2-2\lambda+c_{OPT}}
\]
So either $r_{3}=w_{k}$ works  (but, of course,  may not be the largest value) if $w_k\geq  \frac{\lambda(B-1)}{\lambda^2-2\lambda+c_{OPT}}$, or no value $2\leq r_{3}\leq w_k$  works. 

We now search for $r_{3}\geq M_{*}(p)$. This means we can assume that $OPT_{r_{3}}=w_k$, and the condition becomes $r_{3}-1+B\leq w_{k}(\lambda-1+\frac{1}{\lambda}w_k)$, equivalently $r_{3}\leq 1-B+w_{k}(\lambda-1+\frac{1}{\lambda}w_k)$. To get a value $r_{3}\geq M_{*}(p)=w_k$ we need again that $w_k(\lambda-2+\frac{1}{\lambda}c_{OPT})
\geq B-1$, that is $w_k\geq \frac{\lambda(B-1)}{\lambda^2-2\lambda+w_{k}} $. 
\end{itemize} 
In conclusion, when $r=1$ and $w_k\geq 2$ we have  
\[
r_{3}=\left\{\begin{array}{ll}
1-B+ \lfloor w_k(\lambda-1+\frac{1}{\lambda}w_k) \rfloor  & \mbox{ if } w_k\geq \frac{\lambda(B-1)}{\lambda^2-2\lambda+w_{k}}\\
1 & \mbox{ otherwise.}   
\end{array}
\right. 
\]
Of these two values only the second one is $< M_{*}(p)$. Hence, in the first case $r_3$ does \textbf{not} contribute towards the consistency bound, and   
the consistency bound is 1.  

\item[-] If $2\leq r\leq M_{*}(p)$ then $M_{*}(p)=min(B,r-1+w_k)$, $r_1=r$, $c_{OPT}=1+\frac{w_k-1}{r}$, and 
\[
i_{*}=\left\{\begin{array}{ll} 
1 & \mbox{ if }M_{*}(p)=B  \\
r & \mbox{ otherwise.} 
\end{array} 
\right. 
r_{0}=\left\{\begin{array}{ll} 
1 & \mbox{ if }B<r-1+w_k  \\
r & \mbox{ otherwise.} 
\end{array} 
\right. 
\]

\[
r_2=\left\{\begin{array}{ll} 
\lceil \lambda r\rceil & \mbox{ if }B+ \lceil \lambda r\rceil -1 < r-1+w_k\\
r & \mbox{ otherwise.} 
\end{array} 
\right. 
\]
\[
\frac{P_{r_2}}{M_{*}(p)}=\left\{\begin{array}{ll} 
1+\frac{\lceil \lambda r\rceil -1}{B}
& \mbox{ if }B+ \lceil \lambda r\rceil -1 < r+w_k-1\\
\frac{r+w_k-1}{B}
& \mbox{ if }B+ \lceil \lambda r\rceil -1 > r+w_k-1>B\\
1 & \mbox{ if } M_{*}(p)=r+w_k-1\leq B)
\end{array} 
\right. 
\]
Let's look now for all the $r_3$'s satisfying the condition in the algorithm, so that we can choose the one minimizing ratio $\frac{P_{r_3}}{r_3}$.  There are three cases: 
(a). $1\leq r_3< M_{*}(p)$, $r_3\neq r$. (b). $r_3\geq M_{*}(p)$, $r_3\neq r$. (c). $r_3=r\leq M_{*}(p)$.  Let's deal with these cases in turn: 
\begin{itemize}[leftmargin=*] 
\item[(a).] First, let's look for $1\leq r_3< M_{*}(p)$. We need: 
\[
\frac{P_{r_3}}{r_3}\leq \lambda-1+\frac{1}{\lambda}\cdot c_{OPT}=\lambda-1+\frac{1}{\lambda}\cdot (1+\frac{w_k-1}{r})
\]
For $r_3\neq r$ the condition reads $1+\frac{B-1}{r_3}\leq \lambda-1+\frac{1}{\lambda}\cdot (1+\frac{w_k-1}{r})$. We get 
$M_{*}(p)\geq r_3\geq \frac{(B-1)}{\lambda-1+\frac{1}{\lambda}\cdot (1+\frac{w_k-1}{r})}$, so this case is possible only if $M_{*}(p)[\frac{w_k-1}{r\lambda}+\frac{(\lambda-1)^2}{\lambda}]\geq B-1$, 
 in which case $r_3=M_{*}(p)$, the largest value, minimizes ratio $\frac{P_{r_3}}{r_3}$ among values $\leq M_{*}(p)$. 

\item[(b).] We now look for $r_3\geq M_{*}(p)$ satisfying the condition from the algorithm: 

$\frac{P_{r_3}}{M_{*}(p)}\leq \lambda-1+\frac{1}{\lambda}(1+\frac{w_k-1}{ r})$, or 
$
\frac{r_3-1+B}{M_{*}(p)}\leq \lambda-1+\frac{1}{\lambda}(1+\frac{w_k-1}{r})$, 
yielding 
$M_{*}(p)< r_3\leq 1-B+M_{*}(p)(\lambda-1+\frac{1}{\lambda}(1+\frac{w_k-1}{r}))$. To have such an $r_3$ we need 
$M_{*}(p)[\frac{w_k-1}{r \lambda}+\frac{(\lambda-1)^2}{\lambda}]>B-1$. 
\item[(c).] We also need to check condition for $r_3=r$, that is 
\[
\frac{P_r}{OPT_r}\leq \lambda-1+\frac{1}{\lambda}c_{OPT}= \lambda-1+\frac{1}{\lambda}(1+\frac{w_k-1}{r})
\]
i.e. 
\[
1+\frac{w_k-1}{r}\leq \lambda-1+\frac{1}{\lambda}(1+\frac{w_k-1}{r}), 
\]
that is $
(1+\frac{w_k-1}{r})(\frac{1}{\lambda}-1)+(\lambda-1)\geq 0$. 
which is always true, since it is equivalent to 
$(1-\lambda)[\frac{1}{\lambda}(1+\frac{w_k-1}{r})-1]\geq 0$.
\end{itemize} 

The conclusion is that the largest $r_3$ is either $\geq M_{*}(p)$ or it is equal to $r$, and is the only acceptable value. That is, $r_3$ is equal to: 
\[
\left\{\begin{array}{l} 
r \mbox{ if } M_{*}(p)[\frac{w_k-1}{r\lambda}+\frac{(\lambda-1)^2}{\lambda}]\leq B-1\\

1-B + \lfloor M_{*}(p)(\lambda-1+\frac{1}{\lambda}(1+\frac{w_k-1}{r}))\rfloor \mbox{ otherwise }.  
\end{array} 
\right. 
\]

and $\frac{P_{r_3}}{r_3}$ is equal to 
\[
\left\{\begin{array}{l} 
1+\frac{w_k-1}{r} \mbox{ if }  M_{*}(p)[\frac{w_k-1}{r\lambda}+\frac{(\lambda-1)^2}{\lambda}]\leq B-1\\
1+\frac{B-1}{1-B + \lfloor M_{*}(p)(\lambda-1+\frac{1}{\lambda}(1+\frac{w_k-1}{r}))\rfloor} \mbox{ otherwise. }\\
\end{array} 
\right.
\]
In the second case, though, we have $r_3>  M_{*}(p)$. So in this case $r_3$,  by Proposition~\ref{alg-predictions}, $r_3$ does not contribute to the consistency bound. 

Putting things together, the consistency guarantees for the case $2\leq r\leq M_{*}(p)$ are those described in Table~\ref{sample-table}. 

\item[-] If $M_{*}(p)<r\leq 2B-1$ then $M_{*}(p)=B$, $r_1=r$, $c_{OPT}=\frac{r+w_k-1}{B}$, $i_*=1$, $r_0=1$ (since day $M_*(p)+1$ is not a free day) and 
\[
r_2=\left\{\begin{array}{ll} 
\lceil \lambda r\rceil & \mbox{ if }B+ \lceil \lambda r\rceil -1 < r-1+w_k\\
r & \mbox{ otherwise.} 
\end{array} 
\right. 
\]
\[
\frac{P_{r_2}}{M_{*}(p)}=\left\{\begin{array}{ll} 
1+\frac{\lceil \lambda r\rceil -1}{B} & \mbox{ if }B+ \lceil \lambda r\rceil -1 < r+w_k-1\\
\frac{r+w_k-1}{B} & \mbox{ otherwise.}\\
\end{array} 
\right. 
\]
The condition that $r_3$ has to satisfy is 
$
\frac{P_{r_3}}{OPT_{r_3}}\leq \lambda-1+\frac{r+w_k-1}{\lambda B}$. 
First we check the case $r_3=r$. Since $r>M_{*}(p)$ it follows that $M_{*}(p)=B$ and the condition to satisfy becomes
\[
\frac{r+w_k-1}{B}\leq \lambda-1+\frac{r+w_k-1}{\lambda B}\mbox{ i.e. }
(1-\lambda)(\frac{r+w_k-1}{\lambda B}-1)\geq 0
\]
a condition that is clearly true. So $r_3>M_{*}(p)$, which means that $r_3$ does not contribute to the consistency bound. 

So in this case the consistency guarantee is
\[ 
\left\{\begin{array}{ll} 
1+\frac{\lceil \lambda r\rceil -1}{B} & \mbox{ if }B+ \lceil \lambda r\rceil -1 < r+w_k-1\\
\frac{r+w_k-1}{B} & \mbox{ otherwise.}\\
\end{array} 
\right. 
\]

\end{itemize}



\end{document}